%% file: main.tex
\documentclass[acmsmall]{acmart}

\input{preamble}

\AtBeginDocument{%
  }

\setcopyright{acmlicensed}
\acmJournal{PACMMOD}
\acmYear{2024} \acmVolume{2} \acmNumber{2 (PODS)} \acmArticle{101} \acmMonth{5}\acmDOI{10.1145/3651602}

\begin{document}

\title{Layered List Labeling}

\author{Michael A.~Bender}
\email{bender@cs.stonybrook.edu}
\affiliation{%
  \institution{Stony Brook University and RelationalAI}
  \city{Stony Brook}
  \state{NY}
  \country{USA}
}

\author{Alex Conway}
\email{me@ajhconway.com}
\affiliation{%
  \institution{Cornell Tech}
  \city{New York}
  \state{NY}
  \country{USA}
}

\author{Mart\'in Farach-Colton}
\email{martin@farach-colton.com}
\affiliation{%
  \institution{New York University}
  \city{New York}
  \state{NY}
  \country{USA}
}

\author{Hanna Koml\'os}
\email{hkomlos@gmail.com}
\affiliation{%
  \institution{New York University}
  \city{New York}
  \state{NY}
  \country{USA}
}

\author{William Kuszmaul}
\email{william.kuszmaul@gmail.com}
\affiliation{%
  \institution{Harvard University}
  \city{Cambridge}
  \state{MA}
  \country{USA}
}

\renewcommand{\shortauthors}{Michael A. Bender et al.}

\input{abstract}




\maketitle

\input{intro}

\input{technical}
\input{preliminaries}
\input{composition}
\input{results}

\begin{acks}
We gratefully acknowledge the PODS reviewers, whose helpful and specific comments substantively improved our paper.

This research was partially sponsored by the United States Air Force Research Laboratory and the United States Air Force Artificial Intelligence Accelerator and was accomplished under Cooperative Agreement Number FA8750-19-2-1000. The views and conclusions contained in this document are those of the authors and should not be interpreted as representing the official policies, either expressed or implied, of the United States Air Force or the U.S. Government. The U.S. Government is authorized to reproduce and distribute reprints for Government purposes notwithstanding any copyright notation herein.

This work was also supported by NSF grants 
CCF-2106999, 
CCF-2118620, 
CNS-1938180, 
CCF-2118832,  
CCF-2106827, 
CNS-1938709, and
CCF-2247577. 

Hanna Koml\'os was partially funded by the Graduate Fellowships for STEM Diversity.

Finally, William Kuszmaul was partially supported by a Hertz Fellowship, an NSF GRFP Fellowship and the Harvard Rabin Postdoctoral Fellowship.

\end{acks}

\bibliographystyle{ACM-Reference-Format}
\bibliography{merged}

% Article V2mod086-V2mod126 use
\received{December 2023}
\received[revised]{February 2024}
\received[accepted]{March 2024}

\end{document}

%% file: preamble.tex
\usepackage[utf8]{inputenc}
\usepackage[ruled,vlined,commentsnumbered,titlenotnumbered]{algorithm2e}
\usepackage{url}
\usepackage{enumitem}
\usepackage{mathdots} % provides \iddots
\usepackage{amsmath}
\usepackage{algorithmic}
\usepackage{amsthm}
\usepackage{color}
\usepackage{array}
\usepackage{xy}
\usepackage{setspace}
\usepackage{mathtools}

\usepackage{tikz}
\usetikzlibrary{
  shapes.multipart,
  matrix,
  positioning,
  shapes.callouts,
  shapes.arrows,
  calc}

\usepackage[normalem]{ulem} % allows text to be striked out

\usepackage{hyperref}
\usepackage{cleveref}

\usepackage{accents}

\usepackage{thm-restate}
\usepackage{thmtools}         

\usepackage{graphicx}
\graphicspath{ {./figures/} }

\usepackage{xcolor}

\newcommand{\defn}[1]{\emph{\textbf{{#1}}}}

\newcommand{\E}{\mathbb{E}}

\renewcommand{\paragraph}[1]{\vspace{.3 cm} \noindent \textbf{\boldmath #1} }

\newcommand{\Layer}[2]{#1 \triangleleft #2}
\newcommand{\arr}{\mathcal{A}}

\newcommand{\clb}{c_{\textrm{LB}}}

\declaretheorem[name=Definition]{definition}
\declaretheorem[name=Theorem,sibling=definition]{theorem}
\declaretheorem[name=Lemma,sibling=definition]{lemma}
\declaretheorem[name=Corollary,sibling=definition]{corollary}

\newcommand{\punt}[1]{}

\setlist{itemsep=0pt,topsep=3pt}

%% file: abstract.tex
%!TEX root =  main.tex

\begin{abstract}

%!TEX root =  main.tex

The {list-labeling problem}
is one of the most basic and well-studied algorithmic primitives in data structures, with an extensive literature spanning upper bounds, lower bounds, and data management applications.  The classical algorithm for this problem, dating back to 1981, has amortized cost $O(\log^2 n)$.  Subsequent work has led to improvements in three directions: \emph{low-latency} (worst-case) bounds; \emph{high-throughput} (expected) bounds; and (adaptive) bounds for \emph{important workloads}.

Perhaps surprisingly, these three directions of research have remained almost entirely disjoint---this is because, so far, the techniques that allow for progress in one direction have forced worsening bounds in the others. Thus there would appear to be a tension between worst-case, adaptive, and expected bounds.  
List labeling has been proposed for use in databases at least as early as PODS'99, but a database needs good throughput, response time, and needs to adapt to  common workloads (e.g., bulk loads), and no current list-labeling algorithm achieve good bounds for all three.

We show that this tension is not fundamental. In fact, with the help of new data-structural techniques, one can actually \emph{combine} any three list-labeling solutions in order to cherry-pick the best worst-case, adaptive, and expected bounds from each of them. 

\end{abstract}

%%% Local Variables:
%%% mode: latex
%%% TeX-master: "main.tex" 
%%% End:

%% file: intro.tex
%!TEX root =  main.tex

\section{Introduction}\label{sec:intro}

The \defn{list-labeling problem} is one of the most basic problems in data structures: how can one store a set of $n$ items in \emph{sorted order} in an array of size $(1 + \Theta(1))n$, while supporting both \emph{insertions and deletions}? 
Despite the apparent simplicity of the problem, it has proven remarkably difficult to determine what the best possible solutions should look like. 
Over the past four decades, there has been a great deal of work on upper bounds~\cite{ItaiKoRo81,Willard82, Willard86, Willard92,BenderCoDe02twosimplified, BenderFiGi17,ItaiKa07,AnderssonLa90,GalperinR93,BabkaBCKS19, BenderHu07, Katriel02,BenderFiGi05, BenderDeFa05}, lower bounds~\cite{dietz1990lower,dietz1994tight,dietz2004tight,zhang1993density,Saks18,BulanekKoSa12}, and variations~\cite{BenderHu07,BenderBeJo16,DevannyFiGo17,Dietz82,Andersson89,AnderssonLa90,GalperinR93,Raman99}.

List labeling was proposed for use in database indexing as early as 1999~\cite{Raman99}. 
Today, in the database context, data structures for list labeling are typically called \defn{packed-memory arrays} (PMAs). 
PMAs are used in relational databases~\cite{TokuDB}, NoSql database~\cite{TokuMX}, graph databases~\cite{WheatmanX21, WheatmanX18, WheatmanB21, PandeyWXB21, LeoB21, LeoB19fastconcurrent}, 
cache-oblivious dictionaries~\cite{BenderDeFa00,BenderHuKu16,BenderDuIa02,BenderFiGi05,BenderFK06,BrodalFaJa02}, and order maintenance~\cite{BenderCDF02,BenderCoDe02twosimplified,Dietz82,BenderFiGi17}.

Formally, the  list-labeling problem can be formulated as follows~\cite{ItaiKoRo81}: 

\begin{definition}
A list-labeling instance of \defn{capacity} $n$ stores a dynamic set of up to $n$ elements \emph{in sorted order} in an array of $m=cn$ slots,
for  $c=1+\Theta(1)$. 
Elements are inserted and deleted over time, with each insertion specifying the new element's \defn{rank} $r \in \{1, 2, \ldots, n + 1\}$ among the other elements in the set. 
(Thus, inserting at rank $1$ means that the inserted element is the new smallest element.)

To keep the elements in sorted order in the array, the algorithm must sometimes
move elements around within the array---i.e., \defn{rebalance} elements---e.g., in order to open up a space for a new element.
The \defn{cost} of an algorithm is the number of elements moved during the insertions/deletions.\footnote{To accommodate the many ways in which list labeling is used, some works describe the problem in a more abstract (but equivalent) way: the list-labeling algorithm must dynamically assign each element $x$ a label $\ell(x) \in \{1, 2, \ldots, m\}$ such that $x \prec y \iff \ell(x) < \ell(y)$, and the goal is to minimize the number of elements that are \emph{relabeled} per insertion/deletion---hence the name of the problem.}
\end{definition}

We remark that, as a convention, if the list-labeling algorithm is randomized, then the adversary is assumed to be \defn{oblivious}, meaning that the sequence of insertions and deletions that is performed is independent of the randomness used by the data structure.

\paragraph{$O(\log^2 n)$-cost list labeling and the state of the art.}
In 1981, Itai, Konheim, and Rodeh \cite{ItaiKoRo81} initiated the study of list labeling with a beautiful solution guaranteeing amortized $O(\log^2 n)$ cost per insertion/deletion. This bound would subsequently be independently re-discovered in many different contexts~\cite{Willard81, Andersson89, GalperinR93, Raman99}.

In the four decades since Itai et al's original solution \cite{ItaiKoRo81}, there have been three major ways in which the algorithm has been improved.  These correspond to the three major performance criteria for the use of list labeling in databases: latency, special workload optimizations, and thoughput.

\begin{enumerate}[leftmargin=*]

\item \textbf{Deamortization.}
Itai et al.'s algorithm relies heavily on amortization. There has been a long line of work showing that it is possible to achieve a \emph{worst-case} bound of $O(\log^2 n)$ cost per operation~\cite{Willard82,Willard86,Willard92,BenderCoDe02b,BenderFiGi17}.

\item \textbf{Adaptive Algorithms.} 
The second major direction has been algorithms that \emph{adapt} to the properties of the workload in order to achieve better bounds on natural workloads \cite{BenderFaMo04,BenderFaMo06, BenderHu07, predictions}. This has led to improved bounds both for specific stochastic workloads  \cite{BenderFaMo04,BenderFaMo06, BenderHu07}, and for settings where the list-labeling algorithm is augmented with predictive information \cite{predictions} (i.e., learning with predictions). 

\item  \textbf{Faster Amortized Algorithms.} 
For many years, it was conjectured that, in general, the $O(\log^2 n)$ bound should be optimal~\cite{dietz1990lower, dietz1994tight, dietz2004tight}---and, indeed, lower bounds were established for several classes of algorithms~\cite{dietz1990lower, BulanekKoSa12,BulanekKoSa13}, including any deterministic one \cite{BulanekKoSa12,BulanekKoSa13}. However, it was recently shown that \emph{randomized} algorithms can actually do better, achieving an \emph{expected} cost of $O(\log^{3/2} n)$ per operation~\cite{BenderCoFa22}.
\end{enumerate}

Perhaps surprisingly, these three directions of research have remained almost entirely disjoint---this is because the techniques that allow for progress in one direction tend to lead to \emph{backward progress} in the others.
For example, the role of randomization in the recent $O(\log^{3/2} n)$ algorithm leads to \emph{almost pessimal} tail bounds (the cost is $k$ with probability $\tilde{O}(1/k)$ for any $k \le n$), making deamortization much more difficult.
The known adaptive algorithms \emph{also} rely heavily on amortization; and the approach for achieving $O(\log^{3/2} n)$ relies on a technique from the privacy literature (known as history independence \cite{NaorTe01,Micciancio97,HartlineHoMo02,HartlineHoMo05,BenderBeJo16}), in which one \emph{explicitly commits} to being non-adaptive in one's behavior\footnote{A data structure is said to be \emph{history independent} if, at any given moment, the state of the data structure depends only on the elements that it contains, and not on the history of how they got there.
A key insight in Bender et al.'s $O(\log^{3/2} n)$ algorithm is that history independence can be used to create a barrier between the data structure and the adversary that is using it.
However, in order to enforce this barrier, the data structure must necessarily be non-adaptive.}.

The apparent conflicts between techniques raise a natural question: Can one simultaneously achieve strong results on the three database optimizations of deamortization, adaptivity, and low expected cost? 
We answer this question in the affirmative. In fact, our result is black box: Given three algorithms that achieve guarantees on deamortization, adaptivity, and expected cost, respectively, one can always construct a new algorithm that achieves the best of all three worlds. 

What makes our result interesting is that, intuitively, list-labeling algorithms \emph{should not be composable}. Suppose, for example, that we attempt to interleave two algorithms $X$ and $Y$ so that some elements are logically in $X$, some are logically in $Y$, and all of the elements appear in sorted order in the same array. Whenever a rebalance occurs in $X$, it must \emph{carry around} elements from $Y$ that lie in the same interval. Even if $X$ and $Y$ each individually offer $O(\log^2 n)$ costs, the interleaved algorithm could have arbitrarily poor performance. 

The key contribution of this paper is a more sophisticated approach, in which by treating the problem of composition as a data-structural problem in its own right, we are able to obtain strong black-box results. We emphasize that, although our \emph{techniques} are data structural, the final result is still a list-labeling algorithm: all of the elements appear in sorted relative order in a \emph{single} array of size $(1 + \Theta(1))n$. 

We begin by developing a technique for composing just two list-labeling algorithms, a \defn{fast algorithm} $F$ and a \defn{reliable algorithm} $R$, with the goal of achieving the best properties of both. The new algorithm, denoted by $\Layer{F}{R}$, is referred to as the \defn{embedding} of $F$ into $R$.

\begin{restatable}{theorem}{restatablethmone}
Say that a list-labeling algorithm of capacity $n$ guarantees lightly-amortized expected cost $O(C)$ per operation on an input sequence $\overline{x}$ if, for any contiguous subsequence $\overline{x}_j, \ldots, \overline{x}_{j + T}$ of operations, the total expected cost of the operations is $O(TC + n)$. 
Suppose we are given:
\begin{itemize}[leftmargin=*]
\item A list-labeling algorithm $R$ that has lightly-amortized expected cost $E_R$ per operation and worst-case cost $W_R$ per operation.
\item A list-labeling algorithm $F$ that, on any given operation sequence $\overline{x}$, has amortized expected cost $G_F(\overline{x})$ per operation.
\end{itemize}
Then one can construct a list-labeling algorithm $\Layer{F}{R}$ that satisfies the following cost guarantees:
\begin{itemize}[leftmargin=*]
    \item \defn{Worst-Case Cost.} The worst-case cost of $\Layer{F}{R}$ for any operation is $O(W_R)$.

    \item \defn{Good-Case Cost.} On any input sequence $\overline{x}$, $\Layer{F}{R}$ has amortized expected cost $O(G_F(\overline{x}))$.

    \item \defn{General Cost.} On any input sequence, $\Layer{F}{R}$ has lightly-amortized expected cost $O(E_R)$.

\end{itemize}
Moreover, if $G_F(\overline{x})$ is the same value for all $\overline{x}$, and if $F$'s guarantee is \emph{lightly} amortized (rather than amortized), then the Good-Case Cost guarantee for $\Layer{F}{R}$ is also lightly amortized.
\label{thm:twocomp} 
\end{restatable} 

One should think of the quantities $E_R$ and $W_R$ in Theorem \ref{thm:twocomp} as being functions of $n$ that are known in advance, and one should think of $G_F$ as being a positive-valued function that depends not just on $n$ but also on the input sequence $\overline{x}$.

What makes the specific structure of Theorem \ref{thm:twocomp} powerful (including its somewhat subtle distinction between amortization vs.~\emph{light} amortization) is its repeated \emph{composability}. Given three data structures $X$, $Y$, $Z$, where $X$ has an adaptive guarantee depending on the input, $Y$ has an expected cost guarantee on any input, and $Z$ has a worst-case guarantee on any input, one can apply Theorem \ref{thm:twocomp} twice to conclude that $\Layer{X}{(\Layer{Y}{Z})}$ has \emph{all three properties}.

\begin{restatable}{theorem}{restatablethmtwo}
Consider three list-labeling algorithms $X$, $Y$, $Z$, where $X$ has at most $A(\overline{x})$ amortized expected cost on any input $\overline{x}$, where $Y$ has at most $B$ expected cost on any input, and where $Z$ has worst-case cost at most $C$ on any input. Then, on any input sequence of length $\Omega(n)$, the embedding $$\Layer{X}{(\Layer{Y}{Z})}$$ simultaneously achieves amortized expected cost $O(A(\overline{x}))$ on any input $\overline{x}$; amortized expected cost $O(B)$ on any input; and worst-case cost $O(C)$ on any input.
\label{thm:threecomp}
\end{restatable}

%% file: technical.tex
\paragraph{Technical overview.}
To give intuition for Theorem \ref{thm:twocomp}, let us focus on combining the $G_F(\overline{x})$ input-specific cost of $F$ with the amortized expected $O(E_R)$ cost of $R$ (on any input). This allows us to highlight some of the structural challenges that arise without tackling the problem in its entirety.

As noted earlier, a naive approach to combining $F$ and $R$ would be to have two list-labeling algorithms that are interleaved with one another. Some array slots belong to $F$ and others belong to $R$. When new elements are inserted, they are sent to whichever of $F$ and $R$ can support the insertion more cheaply. At a high level, there are three reasons why this approach ends up failing badly:
\begin{itemize}[leftmargin=*]
\item \textbf{The Deadweight Problem: }In the full array, items must appear in truly sorted order. This forces $F$ and $R$ to be interleaved in potentially strange ways. For example, two elements $f_i$ and $f_{i + 1}$ that appear consecutive to $F$ might have a long sequence $r_j, \ldots, r_k$ of elements from $R$ between them. If $F$ tries to move $f_i$ and $f_{i + 1}$ (at what it thinks is a cost of 2), then it must also carry around $r_j, \ldots, r_k$ as \emph{deadweight} during the move (at an actual cost of $k - j + 3$). 

\item \textbf{The Input-Interference Problem: }Because we selectively choose which of $F$ and $R$ receive each insertion, the specific structure of the input $\overline{x}$ will be corrupted so that the cost $G_F(\overline{x})$ becomes $G_F(\overline{x}')$ for some $\overline{x}'$. Even worse, if either $F$ or $R$ rely on randomization (or adapt to randomization in the input sequence), then that randomization can end up affecting how the input gets partitioned among $F$ and $R$, meaning that the randomization adaptively changes the input! This invalidates any randomized guarantees offered by $F$ or $R$.
\item \textbf{The Imbalance Problem: }If the total number of slots in the array is $1 + \varepsilon n$, and $F$ and $R$ are each allocated $(1 + \varepsilon)n / 2$ slots, then neither data structure can actually fit more than $(1 + \varepsilon)n / 2$ elements. This means that we need to either (1) somehow dynamically change the number of slots allocated to each of $F$ and $R$; or (2) introduce an algorithmic mechanism for keeping $F$ and $R$ load balanced.
\end{itemize}

The first step in resolving these problems is to \emph{embed $F$ into $R$} in a hierarchical fashion, rather than treating the two data structures symmetrically. Now the slots for $F$ (including both empty and occupied slots) are all \emph{elements} of $R$'s array. That is, $R$ views all slots of $F$ (either occupied or free slots in $F$) as occupied slots.
Additionally, $R$ contains some elements (called buffer slots) that $F$ does not know about. If the total number of array slots is $(1 + 3\varepsilon)n$, then one should think of $F$ as occupying $(1 + \varepsilon)n$ slots, and $R$ as getting to make use of the $2\varepsilon n$ additional slots that $F$ does not know about (half of these slots will be used as buffer slots and half will be used as free slots for $R$).

Now, the basic idea is as follows. If an insertion can be implemented efficiently in $F$, then it is placed directly in $F$. Otherwise, if an insertion incurs too much cost in $F$ (more than $\Omega(W_R)$ cost), then the insertion is \defn{buffered} in $R$ \emph{until $F$ can eventually complete the insertion}. The buffering of operations means that $F$ can catch up on rebuild work slowly over time. Finally, whenever $R$ incurs some cost, we also put $\Theta(E_R)$ rebuild work into catching $F$ up. This allows for us to maintain as an invariant that, in expectation, we have put at least as much work into $F$ (over time) as we have into $R$.\footnote{It is tempting to put $\Theta(C)$ rebuild work into $F$, where $C$ is the amount of cost incurred on $R$ during that operation. This type of `direct work matching' would simplify the analysis of the embedding, but it would also subtly reintroduce the input-interference problem, as the quantity $C$ (influenced by $R$'s random bits) would influence the rate at which $F$ catches up, which would influence which operations in the future are buffered/not-buffered, which is what decides $R$'s input.}

Critically, $F$ will eventually perform every insertion/deletion, meaning that, from its perspective, the original input $\overline{x}$ is preserved. Furthermore, although $F$'s behavior affects which insertions get buffered in $R$, the embedding is carefully designed so that the relationship is one-directional. This prevents feedback cycles in which the randomness for $R$ (or $F$) indirectly impacts the future input for $R$ (or $F$), and allows us to avoid the input-interference problem. 

The hierarchical embedding does not help with the deadweight problem, however. When $F$ rearranges what it thinks is a subarray of some size, it may in actuality be carrying around many buffered elements as deadweight. Moreover, the hierarchical embedding would seem to only \emph{worsen} the imbalance problem. 
For example, if the input sequence $\overline{x}$ is such that $G_F(\overline{x}) = \omega(E_R)$, then $F$ may be arbitrarily expensive. 
This means that $F$ will perpetually fall behind $R$, causing the number of buffered elements to balloon uncontrollably.

The second algorithmic insight is that, as elements accumulate in $R$'s buffer, opportunities arise for us to consolidate work performed by $F$. For example, if $F$ plans to rebuild a subarray $A[i, j]$ for the insertion of some element $x$, and then, later on, to rebuild the same subarray again for the insertion of some element $y$, then the two rebuilds can potentially be merged together. By carefully managing the consolidation of work (and how it interacts with the progression of the data structures), we can eliminate both of our remaining problems at once. We resolve the deadweight problem by ensuring that each element $x$ that is inserted is carried around as deadweight at most $O(1)$ total times before it successfully migrates to $F$. And we resolve the Imbalance problem by ensuring that, as $F$'s buffered work accumulates, the work consolidates at a fast enough rate that $F$ is guaranteed to make progress before the number of buffered elements becomes problematically large. 

Although our hierarchical embedding solves the three major problems faced by the naive solution, it requires a great deal of technical care to avoid introducing other new problems. If we are not careful, for example, then whenever we perform work in $F$, the rearrangement of $F$-elements and $R$-elements (moved around as deadweight) will invalidate the state of $R$. Another issue lies in the precise design of Theorem \ref{thm:twocomp}---because we need to be able to apply the theorem twice, we must design its guarantees so that the output of the first application can be fed as a legal input into the second. This leads to several subtleties such as, for example, the role of \emph{light} amortization in the theorem statement. Nonetheless, by designing the embedding in just the right way, we show that it is possible to overcome all of these issues simultaneously, resulting in Theorem \ref{thm:twocomp} and then, by composition, Theorem \ref{thm:threecomp}.

% %%% Local Variables:
% %%% mode: latex
% %%% TeX-master: "main.tex" 
% %%% End:

%% file: preliminaries.tex
\section{Preliminaries}\label{sec:prelims}

In this section, we provide a formal definitions for the list-labeling problem, as well as some mathematical notation we use in the remainder of the paper.

A list-labeling data structure of \defn{capacity} $n$ stores a dynamic set of at most $n$ \defn{elements} in an array of $m=cn$ \defn{slots} for $c=1+\Theta(1)$, maintaining the invariant that the elements appear in sorted order. 
The list-labeling data structure sees the elements that it stores as black boxes---the only information that it knows about the elements is their relative ranks.

The data structure supports \defn{operations} of the form $\overline{x}_t = (r,\sigma)$, where $t$ denotes the timestep of the operation, $\sigma$ specifies whether the operation is an insertion or a deletion, and $r$ is the \defn{rank} of the inserted/deleted element.
An \defn{insert} operation at rank $r$ adds an element with rank $r$ in the data structure, and increments the ranks of all elements whose ranks were at least $r$.
An \defn{delete} operation at rank $r$ deletes the element with rank $r$ from the data structure, and decrements the ranks of all elements whose ranks were at least $r+1$.
We use $\overline{x} = (\overline{x}_1,\dots,\overline{x}_T)$ to denote the input sequence of operations for all timesteps $1,\dots,T$ in the lifespan of the data structure.

In randomized list-labeling data structures, the operations are assumed to be performed by an \defn{oblivious adversary}, who may know the \emph{distribution} of the random decisions made by the data structure, but does not get to see the random decisions made in any specific instance.

The \defn{cost} of a list-labeling algorithm on a sequence of operations is the number of element moves performed by the algorithm on that input sequence. The list labeling algorithm is said to incur \defn{amortized expected cost} at most $O(C)$ if, on every prefix of the input sequence the expected average cost per operation is at most $O(C)$. In this paper, we also define \defn{lightly amortized expected cost}: a list-labeling algorithm has lightly amortized expected cost $O(C)$ on input sequence $\overline{x}$ if on any contiguous subsequence $\overline{x}_j,\dots,\overline{x}_{j+t}$ of $\overline{x}$, the total expected cost on the subsequence is $O(tC+n)$.

%% file: composition.tex
\section{Embedded List-Labeling Algorithms}

In this section, we describe the construction of $\Layer{F}{R}$ (read ``$F$ in $R$''), the embedding of a fast algorithm $F$ into a reliable algorithm $R$. Fix a positive constant $\varepsilon > 0$, and let $F$ and $R$ be two list-labeling algorithms, where $F$ is on an array of size $(1 + \varepsilon)n$ capable of holding up to $n$ elements, and $R$ is on an array of size $(1 + 3\varepsilon)n$ capable of holding up to $(1 + 2 \varepsilon)n$ elements. The full embedding $\Layer{F}{R}$ will be an array $\arr$  of size $(1 + 3\varepsilon)n$ capable of holding up to $n$ elements.\footnote{In order to achieve a final slack of $\varepsilon$ in the overall embedding, we would need to use $\varepsilon/3$ here. For clarity of notation, we use $\varepsilon$, and achieve a slack of $3\varepsilon$ in the resulting algorithm.}

 The embedding $\Layer{F}{R}$ consists of two primary components, a list-labeling algorithm called the \defn{F-emulator}, which is implemented using (a modified version of) $F$, and a list-labeling algorithm called the \defn{R-shell}, which is implemented using $R$.  The $F$-emulator runs on a subarray denoted by $\arr_F$ of size $n (1 + \varepsilon)$ (although, as we shall see, the choice of which slots comprise this subarray will change over time), and the $R$-shell runs on the entire array, including the remaining $2 \varepsilon n$ slots of $\arr \setminus \arr_F$. An example of what $\arr_F$ looks like is given in Figure \ref{fig:slot_types} (we will say more about this figure later).

\paragraph{High-level roles of the $F$-emulator and $R$-shell.} The $F$-emulator maintains a simulated copy of $F$, i.e., it keeps track internally of what the state of $F$ would be on an array of size $(1 + \varepsilon)n$ elements. As we shall see, this simulated copy of $F$ does not physically exist (i.e., it is not necessarily what the array $\arr_F$ stores at any given moment). The simulation is just to help the $F$-emulator with planning.

The $F$-emulator works in batches to transform the state of $\arr_F$ into that of the simulated copy---each batch aims to get the $F$-emulator to the same state that the $F$-simulator was in \emph{when the batch began}. We refer to the ongoing process of transforming the state of $\arr_F$ into the state of the simulated copy of $F$ as a \defn{rebuild}, and at a given time step we say there is a \defn{pending rebuild} if the $F$-emulator is not fully caught up with the simulated $F$. In order that the $F$-emulator is not too expensive on any given time step, rebuilds perform at most $O(E_R)$ work per time step.

If, when insertion occurs, there is already a pending rebuild, or if inserting that item would \emph{cause} a pending rebuild that requires more than $\Omega(E_R)$ work, then the inserted item must be \defn{buffered} in the $R$-shell. This means that the item is stored in one of the slots $\mathcal{A} \setminus \mathcal{A}_F$ that the $R$-shell knows about but that the $F$-emulator does not. As we shall see, the item is only moved into the $F$-emulator once the $F$-emulator eventually catches up to a state where it knows about that item. That is, the item will be moved from a buffer slot into $\mathcal{A}_F$ once the $F$-emulator performs a rebuild that transforms $\mathcal{A}_F$ into a state that contains the item. 

Whereas the $F$-emulator must maintain a simulated copy of $F$, the $R$-shell will not need to do any such thing. Rather, the $R$-shell will directly use $R$ to implement the sequence of operations that it receives (although, as we shall see, this sequence is not the same as the original input sequence given to $\Layer{F}{R}$).

\begin{figure}
\centering
\includegraphics[width=0.9\linewidth]{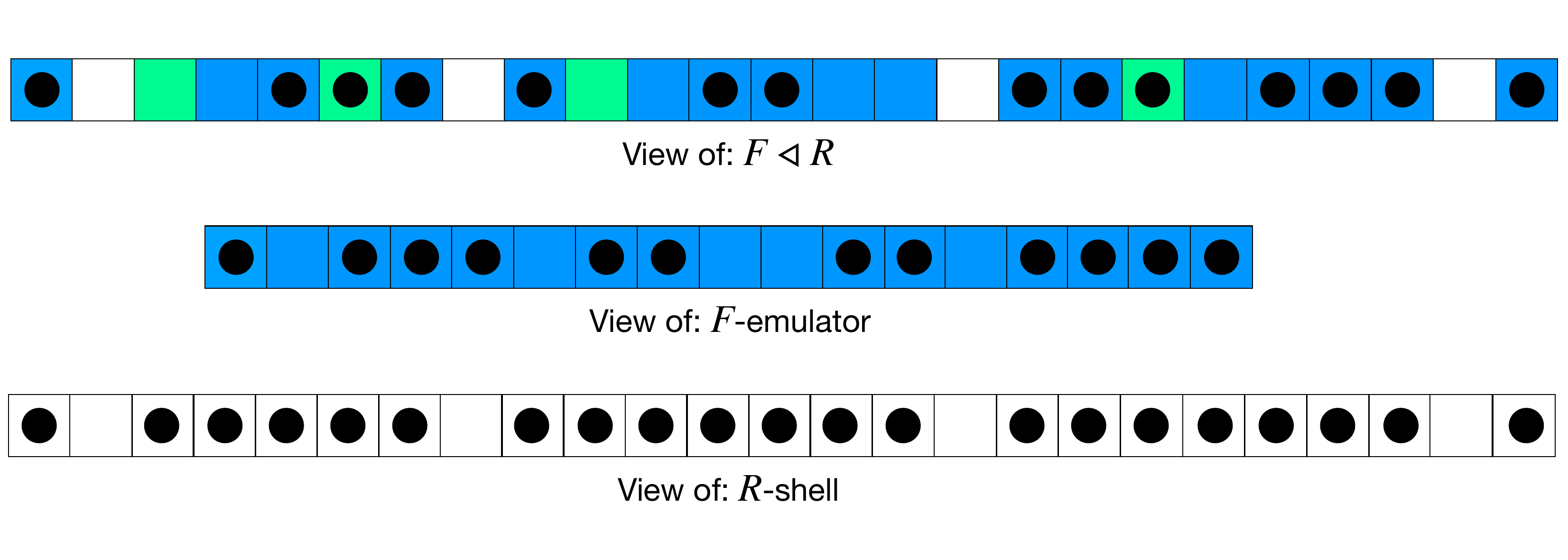}
\caption{An example array $\arr$. The first image shows the data structure from the view of the embedding $\Layer{F}{R}$. There are 17 $F$-emulator slots, shaded blue, of which 12 are occupied by real elements. There are 4 ($R$-shell) buffer slots, shaded green, of which 2 are occupied by real elements. Finally, there are 4 $R$-shell empty slots, shaded white. The second image shows the data structure from the view of the $F$-emulator (i.e., the array $\mathcal{A}_F$), which only sees the blue slots. The third image shows the view of the $R$-shell, which is aware of all slots in the array, but sees all $F$-emulator slots (occupied and free) and buffer slots (occupied and free) as occupied by elements.}
\label{fig:slot_types}
\end{figure}

\paragraph{Types of slots.}
Figure \ref{fig:slot_types} gives an example of different types of slots that can exist in $\Layer{F}{R}$ at any given moment. These include: the \defn{$F$-emulator slots} (i.e., $\mathcal{A}_F$); $\epsilon n$ \defn{($R$-shell) buffer slots}; and $\epsilon n$ \defn{$R$-shell empty slots}.  

From the perspective of the $F$-emulator, the $F$-emulator slots are the only slots that exist, consisting of both the items and free slots for the $F$-emulator. From the perspective of the $R$-shell, the $F$-emulator slots are all occupied slots. Additionally, from the perspective of the $R$-shell, the buffer slots are \emph{also} all occupied slots. If an $F$-emulator slot (resp.~a buffer slot) does not actually contain an item, then the $R$-shell treats it as containing an \defn{$F$-emulator dummy element} (resp.~a \defn{buffer dummy element}). The only free slots, from the perspective of the $R$-shell, are the $\epsilon n$ $R$-shell empty slots. The other slots are viewed by the $R$-shell as elements that it can move around.

\paragraph{How moves in the $F$-emulator are implemented in $\Layer{F}{R}$.} An important issue that we will need to be careful about is \emph{cost amplification} in the $F$-emulator: when we rearrange items in the $F$-emulator, we will need to \emph{also} rearrange elements in $R$-shell whose ranks lie between those that we rearranged in the $F$-emulator. 

Suppose that the $F$-emulator wishes to move an element $x$ into an ($F$-emulator) free slot $s$ immediately to $x$'s right (in the $F$-emulator). Suppose, furthermore, that there are $a$ buffer slots between $x$ and $s$, $a_1$ of which contain actual elements and $a_2 > 0$ of which contain dummy elements. Let $i_1 < i_2 < \cdots < i_{a + 2}$ denote the positions in which $x$, the $a$ buffer slots, and $s$ appear, respectively. 
We cannot move $x$ directly into slot $s$, because it would jump over $a_2$ actual elements. Instead, the embedding $\Layer{F}{R}$ must move the $a$ buffer slots from positions $i_2, \ldots, i_{a + 1}$ to positions $i_3, \ldots, i_{a + 2}$, respectively; this creates a free slot in position $i_2$, which the embedding moves $x$ into; finally, the embedding reclassifies position $i_2$ as an $F$-emulator slot (i.e., placing it in $\mathcal{A}_F$) and reclassifies position $i_{a + 2}$ as a buffer slot (i.e., removing it from $\mathcal{A}_F$). 
An example of this process is shown in Figure \ref{fig:slot_moves}.

From the perspective of the $F$-emulator, we have moved $x$ into slot $s$. From the perspective of the $R$-shell, we have done \emph{nothing}, as the set of occupied slots (i.e., $F$-emulator slots and buffer slots) has remained unchanged. 
Critically, although we have changed which slots comprise $\mathcal{A}_F$, we have \emph{not} changed which slots the $R$-shell views as occupied. Importantly, the $R$-shell, being a list-labeling algorithm, does not care about what is actually stored in a given slot: the $R$-shell's behavior is \emph{completely determined} by (1) which slots it thinks are occupied and (2) what ranks it is asked to perform insertions/deletions at. 

In general, if we wish to move an $F$-emulator item $x$ to an $F$-emulator free slot $s$ that is $i$ positions to $x$'s left/right in the $F$-emulator, then we can achieve this by repeatedly applying the construction above. The total cost is $O(1 + a_1)$, where $a_1$ is the number of (non-dummy) items in buffer slots between $x$ and $s$. Thus, a rearrangement that the $F$-emulator thought should cost $O(1)$ actually costs $O(1 + a_1)$ due to \defn{cost amplification}. We refer to the $O(a_1)$ extra moves that needed to be performed as \defn{deadweight moves}---bounding the cost of these moves will be a critical piece of our analysis. 

We remark that, except when specified otherwise, whenever we refer to the cost incurred during a rebuild, we will \emph{include} the cost of deadweight moves. On the other hand, when referring to costs incurred by the simulated copy of $F$, we do not include amplification costs, since the simulated copy of $F$ does not incur deadweight moves.

\begin{figure}
\centering
\includegraphics[width=0.45\linewidth]{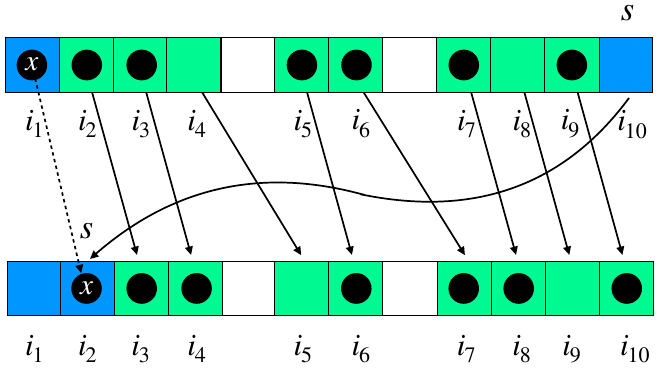}
\caption{An example move in the $F$-emulator of element $x$ to a ($F$-emulator) neighboring free slot $s$. Here $a=8$ buffer slot elements sit in between $x$ and the $F$-free slot in $\Layer{F}{R}$, with $a_1=6$ containing real buffered elements. The remaining $a_2=2$ contain dummy elements. Solid lines represent moves of slots in $\Layer{F}{R}$, while the dashed line represents the move of the element $x$ in the $F$-emulator. From the view of the $F$-emulator, all that has happened is that $x$ moved into slot $s$; and from the view of the $R$-emulator, nothing has happened.}
\label{fig:slot_moves}
\end{figure}

\paragraph{Implementation of the $F$-emulator.}
We are now prepared to describe the implementation of the $F$-emulator. 
As its internal state, the $F$-emulator keeps track of a simulated copy of $F$ on an array of size $(1 + \varepsilon) n$. The actual state of the $\arr_F$ may not match the state of the $F$-emulator at any given moment. To distinguish between them, we will use $F(t)$ to denote the state of the simulated copy of $F$ after the $t$-th operation, and $\widetilde{F}(t)$ to denote the state of $\arr_F$ after the $t$-th operation. 

At a high level, the goal of the $F$-emulator will be to gradually perform work on $\widetilde{F}(t-1)$ with the goal of bringing its state closer to that of $F(t)$.
Because $F(t)$ changes after each time step, it is convenient to freeze the version of $F$ that the $F$-emulator targets, which we call the \defn{checkpoint}. Checkpoints are useful to ensure that progress is made at each time step, because the target state is unchanged until the transformation into the checkpoint is complete.

When a \defn{rebuild} begins, at some time $t_0$, then the state $F(t_0)$ will be the checkpoint used for that rebuild. For all times $t \ge t_0$ until the rebuild finishes, we define the \defn{target checkpoint state} $C(t) = F(t_0)$. In the same way that the simulated copy of $F$ is stored internally by the $F$-emulator, the target state $C(t)$ is also stored internally.

When a rebuild finishes at some time $t_1$, the state $F(t_1)$ may be quite different from $F(t_0)$.  At this point, the next checkpoint time is set to be $t_1$, and $F(t_1)$ becomes the target of the next rebuild.

\begin{figure}
\centering
\includegraphics[width=0.7\linewidth]{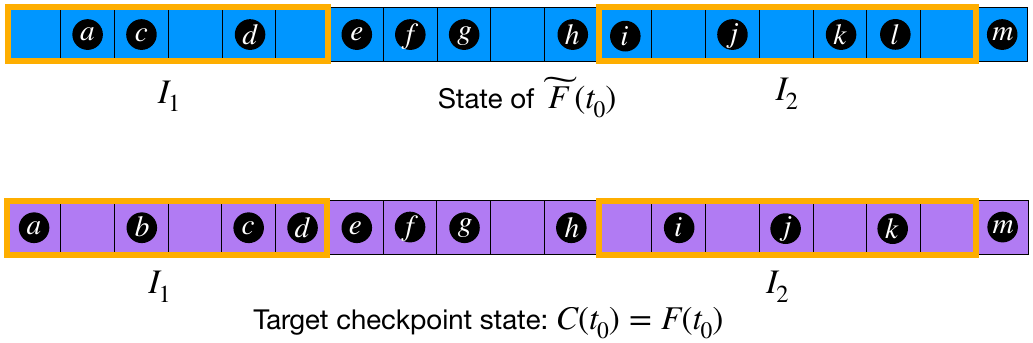}
\caption{A example of the intervals $I_1, I_2, \ldots, I_k$ (where $k = 2$) used by a rebuild beginning at time $t_0$. The states of $\widetilde{F}(t_0)$ (i.e., the slots in $\arr_F$) and $F(t_0)$ (i.e., the simulated copy of $F$) are each shown, and the intervals $I_1$ and $I_2$ are the constructed based on which elements need to move in order to get from one state to the other. The elements that need to move ($a, b, c, d, i, j, k, \ell$) form the set $Q$, and $I_1$ and $I_2$ are defined to be the maximal sub-intervals out of those that contain just elements of $Q$ and that are non-empty.}
\label{fig:rebuild_intervals}
\end{figure}

The rebuild starting at time $t_0$ which transforms $\widetilde{F}(t_0)$ to $C(t_0) = F(t_0)$ is accomplished as follows. Recall that each of $\widetilde{F}(t_0)$ and $F(t_0)$ are arrays of size $(1 + \epsilon)n$ that contain (up to) $n$ elements. Let $Q$ be the set of elements that need to be moved (or inserted or deleted) in order to transform $\widetilde{F}(t_0)$ into $F(t_0)$, that is, the set of elements that appear in at least one of $\widetilde{F}(t_0)$ and $F(t_0)$ but that do not appear in the same array slot in both. Let $I_1,\dots, I_k$ be the set of maximal non-empty contiguous intervals of the $F$-emulator's array $\arr_F$ that contain only elements of $Q$ (see Figure \ref{fig:rebuild_intervals}). 
For each subinterval $I_j$, the rebuild needs to rearrange/insert/delete the elements within the interval to get from their state in $\widetilde{F}(t_0)$ to their state in $F(t_0)$. 

\begin{figure}
\centering
\includegraphics[width=0.55\linewidth]{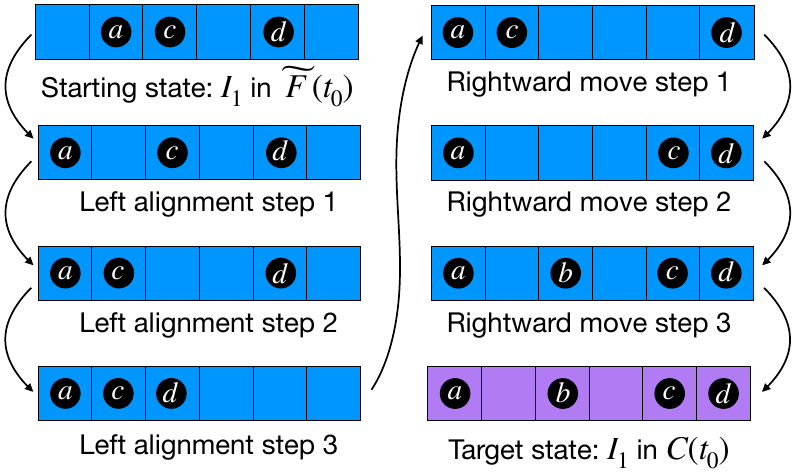}
\caption{An example of rebuilding the interval $I_1$, in Figure \ref{fig:rebuild_intervals}, by first moving the elements in the interval to be left-aligned, and then to their correct positions within $\arr_F$. Each step shows only the state of $I_1$, which in turn is a sub-interval of $\arr_F$, so slots not in $\arr_F$ (i.e., slots colored green and white in Figure \ref{fig:slot_types}) are not shown (this means that deadweight moves are also not shown).  Starting at the state of the interval in $\widetilde{F}(t_0)$, the rebuild first moves the elements in the interval one-by-one to be left-aligned in the interval. The rebuild then moves the elements one-by-one to their target positions within the array $\arr_F$ (i.e., their positions in $F(t_0)$, which is also $C(t)$ for every time $t$ within the rebuild time window). Note that Rightward-move step 3 is an \emph{incorporation step}, moving an element that was formerly in an R-shell buffer slot (so not formerly in $\arr_F$) into a slot within $\arr_F$. Also note that the final rightward-step (which would be Rightward move step 4) is a no-op, since $a$ is already in its correct position within array $\arr_f$ and does not need to be moved.}
\label{fig:interval_moves}
\end{figure}

Now let us describe how a given rebuild rearranges the elements within a given subinterval $I_j$ (Figure \ref{fig:interval_moves}). This process takes place within the array $\arr_F$, so, as discussed earlier, buffered elements (i.e., occupied green slots in Figure \ref{fig:slot_types}) may be carried around via deadweight moves. To rearrange the elements within a given subinterval $I_j$, the rebuild first moves all elements to be completely left-aligned within the subinterval (again, this rearrangement occurs in $\arr_F$); and the rebuild then moves all elements (rightwards) into their correct places, one at a time, starting with the rightmost element. Critically, this latter step also incorporates any elements from the buffer slots that were not present in $\widetilde{F}(t_0)$ but that are present in $C(t) = F(t_0)$ (i.e., these elements move from buffer slots to $\mathcal{A}_F$ slots). 
The two-phase approach of first moving elements to be left-aligned, and then moving them to their correct positions, ensures that the array is always in a correct state with respect to element ranks, and at most doubles any associated costs.

It is worth taking a moment to comment on what it means to incorporate an element from a buffer slot into $\mathcal{A}_F$. As the rebuild occurs on the $F$-emulator, some of the elements in buffer slots (specifically, those that are present in $C(t) = F(t_0)$) are moved from buffer slots into $\mathcal{A}_F$. Whenever an element is incorporated, it is moved from some buffer slot to some $F$-emulator slot. In this case, the buffer slot \emph{remains} a buffer slot, and is now said to contain a buffer dummy element.

\paragraph{Insertions in $\Layer{F}{R}$.}
We are now prepared to describe the full implementation of $\Layer{F}{R}$. Consider an insertion $x$. Let $c_E$ be the cost incurred by the $F$-emulator's simulated copy of $F$ to insert $x$. There are two cases:
\begin{enumerate}[leftmargin=*]
   \item The \defn{fast path} occurs if there is no pending rebuild and $c_E \leq E_R$. In this case:
   \begin{enumerate}
       \item Insert $x$ into $\arr_F$ by emulating $F$ as in the $F$-emulator's simulated copy.
   \end{enumerate}
   \item The \defn{slow path} occurs otherwise. In this case:
   \begin{enumerate}
       \item Insert $x$ into the $R$-shell:
       \begin{enumerate}
           \item Choose an arbitrary buffer slot that currently contains a dummy element, and delete it using $R$.
           \item Insert a new buffer slot at $x$'s rank using $R$.
           \item Put $x$ into the new buffer slot.
       \end{enumerate}
       \item Perform rebuild work in the $F$-emulator:
       \begin{enumerate}
           \item Perform $\Theta(E_R)$ rebuild work in the $F$-emulator.
           \item If, additionally, the rebuild can be completed at cost less than $E_R$, do so.
           \item If the rebuild is complete, set the new checkpoint to be the state of the $F$-emulator's simulated copy of $F$.
           \item If the [new] rebuild can be completed at cost less than $E_R$, do so.
       \end{enumerate}
   \end{enumerate}
\end{enumerate}

We emphasize that in part (b) of the slow path, during steps (i), (ii), and (iv), every time that we refer to the cost of rebuild work, we are \emph{including} the cost of deadweight moves. 

It is also worth remarking on the role of steps (ii) and (iv) in part (b) of the slow path. This is to ensure that, whenever there is a pending rebuild, there is $\Omega(E_R)$ pending work to be done. This is important for step (i) of part (b) of the slow path, which instructs the $F$-emulator to perform $\Theta(E_R)$ rebuild work.

Finally, it should be noted that, at the beginning of time, $R$ must be initialized to contain $\mathcal{A}_F$, along with the $\epsilon n$ buffer slots. This requires the simulation of $\Theta(n)$ insertions on $R$. This is also why we require $R$ to offer a \emph{lightly} amortized guarantee (rather than simply an amortized one), since the lightness of the amortization allows for us to bound the effect of these initialization insertions on the cost of later operations in $R$.

\paragraph{Deletions in $\Layer{F}{R}$.} To delete an item $x$, our first step is to simply remove the item. If the item being deleted is in a buffer slot, then the slot remains a buffer slot (now containing a dummy element). Similarly, if the item being deleted is in an $F$-emulator slot, then that slot also remains an $F$-emulator slot. The $F$-emulator will treat that slot as containing the deleted element, up until the emulator has caught up to a point where the deletion has occurred.

After removing the item, we proceed with almost the same logic as for insertions---the only difference is that now we can skip part (a) of the slow path. In more detail, we let $c_E$ be the cost incurred by the $F$-emulator's simulated copy of $F$ to delete $x$, and there are two cases:
\begin{enumerate}[leftmargin=*]
   \item The \defn{fast path} occurs if there is no pending rebuild and $c_E \leq E_R$. In this case:
   \begin{enumerate}
       \item Delete $x$ from $\arr_F$ by emulating $F$ as in the $F$-emulator's simulated copy.
   \end{enumerate}
   \item The \defn{slow path} occurs otherwise. In this case, we perform part (b) of the slow path for insertion.
   \end{enumerate}

A nice feature of deletions is that they will not require any special handling in our analyses. This will be because (1) the fact that an item has been \emph{removed} will only reduce the costs that we are calculating, and (2) all slow-path operations (both insertions and deletions) always perform at least $\Theta(E_R)$ work on the current rebuild.

%% file: results.tex
\section{Proof of Theorems \ref{thm:twocomp} and \ref{thm:threecomp}}

We now analyze the embedding in order to establish Theorems \ref{thm:twocomp} and \ref{thm:threecomp}. Recall from the technical overview that there are three major issues we must avoid: the deadweight problem, the input-interference problem, and the imbalance problem. 

We remark that the imbalance problem, in particular, is is a matter of \emph{well-defined-ness} for our insertion algorithm -- we must show that, whenever the slow path is invoked by an insertion, there exists at least one available buffer slot (i.e., a buffer slot that contains a dummy element). We will prove this in Lemma \ref{lem:buffer-slots-free}. In order that, in the lemmas preceding Lemma \ref{lem:buffer-slots-free}, our discussion is well-defined, we shall add for now a \defn{halting condition} to our algorithm: if, at any point in time, there are no buffer slots available to handle an insertion, then the algorithm halts (no more operations are performed), and the remaining operations are treated as having cost zero. Of course, we will ultimately see via Lemma \ref{lem:buffer-slots-free} that this halting condition can never occur.

We begin by establishing formally that the design of the embedding avoids the input-interference problem.

 Recall that $F$ and $R$ are each (potentially) randomized algorithms. Let $\operatorname{rand}(F)$ and $\operatorname{rand}(R)$ be the strings of random bits used by each of the two algorithms. Additionally, let $\overline{x}$ denote the input received by $F$ (i.e., the simulation of $F$ maintained by the $F$-emulator), and let $\overline{y}$ be the input received by the $R$ (i.e., the $R$-shell). Both inputs are sequences of insertions and deletions, where each operation specifies a rank at which the insertion or deletion should occur. Whereas $\overline{x}$ is the same input that the full embedding $\Layer{F}{R}$ receives, $\overline{y}$ is determined by a more complicated algorithmic process. We shall now establish that $\overline{y}$ is fully determined by $\overline{x}$ and $\operatorname{rand}(F)$, and is therefore independent of $\operatorname{rand}(R)$. That is, the random bits for $R$ are independent of its input.
\begin{lemma}
The sequence  $\overline{y}$ of operations given to the $R$-shell is independent of the $R$-shell's random bits  $\operatorname{rand}(R)$.\label{lem:nocycles}
\end{lemma}
\begin{proof}
 Concretely, we will show that $\overline{y}$ is the same, no matter how the $R$-shell behaves. That is, $\overline{y}$ is fully determined by $\overline{x}$ and $\operatorname{rand}(F)$. 

At any given moment, define the \defn{truncated state} $T$ of $\Layer{F}{R}$ to be the state of the array $\Layer{F}{R}$ except with the $R$-empty slots removed (i.e., the slots of $T$ are the green and blue slots in Figure \ref{fig:slot_types}), and with each remaining slot annotated to indicate whether it is green or blue in Figure \ref{fig:slot_types}. When the $R$-shell moves elements around, the state of $T$ does not change (i.e., all that moves by the $R$-shell change is the interleaving of the white slots in Figure \ref{fig:slot_types} with the green and blue slots). Since the $F$-emulator's behavior is oblivious to where the $R$-shell free slots are, we can think of the $F$-emulator as interacting directly with $T$. That is, even if we do not give the $F$-emulator access to the full state of $\Layer{F}{R}$, but instead only access to $T$, we can fully reconstruct how the $F$-emulator behaves over time. Moreover, since the state of $T$ is unaffected by $R$'s behavior, the evolution of $T$ over time (i.e., what it looks like after each operation in the original input sequence $\overline{x}$) is fully determined by the original input sequence $\overline{x}$ and the random bits $\operatorname{rand}(F)$. Thus the state of $T$ and the behavior of the $F$-emulator (including the choice of which elements are inserted/deleted via fast vs slow paths) is fully determined by the original input sequence $\overline{x}$ and the random bits $\operatorname{rand}(F)$. However, the state of $T$ along with the indicator random variables for which elements in $\overline{x}$ are inserted/deleted via fast vs slow paths fully determine the input sequence $\overline{y}$ that is given to the $R$-shell. Thus the input to the $R$-shell is independent of $\operatorname{rand}(R)$, as desired.

\end{proof}

Next, we address the deadweight problem. We show that, although elements can be moved around (as deadweight), each element is only involved in $O(1)$ total such moves.
\begin{lemma}\label{lem:aux-moves}
    Each item is moved by at most $4$ deadweight moves. Moreover, these deadweight moves occur during either the rebuild in which the item was inserted, or during the next rebuild, and each rebuild moves the item as a deadweight move at most twice.
\end{lemma}

\begin{proof}
    Consider an insertion $x$ at some time $t$. Suppose that at time $t$ we are performing a rebuild on the $F$-emulator that brings the $F$-emulator to checkpoint state $C(t) = F(t_0)$ for some $t_0 \le t$. The next rebuild will bring the $F$-emulator to state $F(t_1)$ for some $t_1 \ge t$. Critically, the checkpoint $F(t_1)$ is after $x$'s insertion time. Thus, by the time the next rebuild is complete, the inserted item $x$ will be incorporated into the $F$-emulator (or $x$ will have been deleted). So the only opportunities for $x$ to be involved in a deadweight move are during the current rebuild or the next one.
    
    To complete the proof, it suffices to show that, during a given rebuild, each item $x$ (that is in a buffer slot) is involved in at most $2$ deadweight moves. Recall that, during a rebuild, the $F$-emulator decomposes $\mathcal{A}_F$ into disjoint intervals $I_1, \ldots, I_k$ (for some $k$) and performs two passes on each $I_j$ (one to move all elements to be completely left aligned, and one to move them into their correct places). The item $x$ incurs deadweight moves from at most one $I_j$, and therefore incurs at most $2$ such deadweight moves  during a given rebuild.
\end{proof}

The key to avoiding the imbalance problem is to prove that each rebuild spans a relatively short period of time.
\begin{lemma}\label{lem:checkpoint-length}
Supposing $n$ is at least a sufficiently large positive constant, there exists a positive constant $c$ such that each rebuild completes in at most $c n / \log n = o(n)$ operations. 
\end{lemma}

\begin{proof}
We will prove the lemma by induction on the number of rebuilds that have occurred so far.

Let $S_1$ and $S_2$ be the number of operations that occur during the prior rebuild and the current rebuild, respectively.
Since the $F$-emulator performs $\Theta(E_r)$ rebuild work during every operation (possibly except the final one) in the rebuild, and since $E_r = \Omega(\log n)$ by the lower bound of \cite{BulanekKoSa13}, we have that the total cost of the rebuild is at least 
\begin{equation} \clb \cdot (S_2 - 1) \cdot \log n \label{eq:cprime}\end{equation} for some positive constant $\clb$.
Set $c = \max{(4,16/\clb)}$, and assume that $n$ is large enough that both $\log{n} \geq 8/\clb$ and $n \geq 2$.
Our inductive hypothesis will be that $S_1 \leq cn/\log{n}$.
The base case is the first rebuild, in which case $S_1 = 0$.

Now, suppose that a rebuild starts at time $t_0$, and that the inductive hypothesis holds.
We obtain an upper bound on the cost of the rebuild as follows.
\emph{Excluding deadweight moves}, the cost of the rebuild is at most $4n$, since the non-deadweight moves are spent moving around elements that are present in at least one of $\widetilde{F}(t_0)$ and $C(t_0)$, and each such element is moved at most twice during the rebuild. 
On the other hand, by Lemma \ref{lem:aux-moves} the total cost of deadweight moves during the rebuild is at most $2 (S_1 + S_2)$. 
By the inductive hypothesis, we know that $S_1 \le c n / \log n$. So the total cost of the current rebuild, including cost amplification is at most

\begin{equation} 4n + 2 S_1 + 2 S_2 \le 4n + 2c n / \log n + 2S_2. \label{eq:checkpointupper}\end{equation}

Combining \eqref{eq:checkpointupper} and \eqref{eq:cprime}, we have that
$$\clb \cdot (S_2 - 1) \cdot \log n \le 4n + 2c n / \log n + 2S_2.$$

Rearranging terms, we have that
\[ S_2 \leq \frac{1}{\clb\log n -2}\left( 4n + \frac{2c n}{\log n} + \clb \log n \right),\]
and we want to show that the right-hand side is at most $cn/\log n$.
Rearranging terms again, this is equivalent to showing:
\begin{align}\label{eq:checkpoint-wts}
    \clb\log{n} + 4n &\leq \frac{cn}{\log{n}}\left(\clb\log{n} - 2\right) - \frac{2cn}{\log{n}} 
    =cn\left(\clb - \frac{4}{\log{n}}\right).
\end{align}

By assumption, $\log{n} \geq 8/\clb$, so $4/\log{n} \leq \clb/2$.
Therefore,
\[\frac{1}{2}cn\left(\clb - \frac{4}{\log{n}}\right) \geq c\cdot\clb n/4.\]
On the one hand, this is at least $4n$, since $c \geq 16/\clb$ by our choice of $c$.
On the other hand, this is at least $\clb \log{n}$, since $c \geq 4$ (again by choice) and $n \geq \log{n}$ (since $n \geq 2$).
Combining the two halves, we have shown \eqref{eq:checkpoint-wts}, which completes the induction.
\end{proof}

We can now show that the imbalance problem is avoided. That is, whenever the $R$-shell needs a buffer slot to place an element in, there is at least one available, so the halting condition specified at the beginning of the section never occurs.
\begin{lemma}
There are $o(n)$ buffer slots in use at any time. In particular, whenever step (i) of part (a) of the slow path is invoked, there exists at least one buffer slot that contains a dummy element. Thus the halting condition never occurs. \label{lem:buffer-slots-free}
\end{lemma}

\begin{proof}
Consider any time $t$.
If the $F$-emulator is not in the process of executing a rebuild at time $t$, then there are no items in the $R$-shell, and there is nothing to prove.
Any insertions that occurred before the previous checkpoint have already been incorporated into the $F$-emulator by construction.
Therefore, all of the items in the $R$-shell must have been inserted during the current rebuild or the previous one.
By \Cref{lem:checkpoint-length}, there are $o(n)$ such insertions.
\end{proof}

This brings us to the task of actually bounding the \emph{cost} incurred by the embedding. This is the most subtle part of the analysis.

Recall that a list-labeling algorithm with capacity $n$ is said to guarantee \defn{lightly-amortized expected} cost $C$, if on any operation sequence $\overline{x}_t$, the expected cost incurred by any subsequence of operations $\overline{x}_a, \ldots, \overline{x}_{a + i - 1}$ is at most $i\cdot C + O(n)$. 
We will now prove that, if $R$ has lightly-amortized expected cost $O(E_R$), and $F$ has amortized expected cost $O(G_F)$, then $\Layer{F}{R}$ incurs amortized expected cost $O(G_F)$.

\begin{lemma}
Suppose that $R$ has lightly-amortized expected cost $O(E_R)$ per operation; and that $F$ has amortized expected cost $O(G_F)$ per operation (where $G_F$ may be a function of the input sequence). Then, on any input sequence $\overline{x}$ of length $\Omega(n)$, the embedding $\Layer{F}{R}$ incurs amortized expected $O(G_F)$ cost per operation. 
\label{lem:EF1}
\end{lemma}
\begin{proof}
Consider an input sequence $\overline{x}_t$ of length $T \ge \Omega(n)$. Let $t_1, t_2, \ldots, t_j$ be the operations that trigger the slow path in the embedding. Each $t_i$ triggers $2$ operations (an insertion and a deletion) in the $R$-shell, leading to a total of $2j$ operations in the $R$-shell. 

By Lemma \ref{lem:nocycles}, the operations given to the $R$-shell are independent of the $R$-shell's random bits, so we can apply $R$'s light-amortization guarantee to bound the expected cost incurred by the $R$-shell on the $2j$ operations by
\begin{equation}
    O(j E_R + n).
    \label{eq:cost2j}
\end{equation}  

It is worth emphasizing why we needed a \emph{light} amortization guarantee from $R$ in order to arrive at \eqref{eq:cost2j}. The issue is that, when the embedding is initialized, the $R$-shell must be initialized to contain $\Theta(n)$ dummy elements, meaning that the $R$-shell has already has $\Theta(n)$ (virtual) operations by the time the input sequence starts. Thus the $2j$ operations that the $R$-shell incurs after initialization must be treated as an \emph{subsequence} of $R$'s operation sequence---this is why we need $R$'s guarantee to be \emph{lightly} amortized.

On the other hand, each of the $j$ operations that trigger the slow path perform $\Theta(E_R)$ rebuild work on the $F$-emulator. This means that the $F$-emulator incurs cost at least $\Omega(jE_R)$. It follows that, if $C_1$ is the total expected cost incurred on the $R$-shell, and $C_2$ is the total expected cost incurred on the $F$-emulator (including deadweight moves), then 
\begin{equation}
C_1 \le O(j E_R + n) \le O(C_2 + n),
\label{eq:C1toC2}
\end{equation}
where the first inequality follows from \eqref{eq:cost2j}. Since the input sequence has length $T \ge 
 \Omega(n)$, the amortized expected cost per operation is at most
\begin{equation}O((C_1 + C_2) / T) \le O(C_2/T + n/T) = O(C_2 / T + 1),\label{eq:c2n1}\end{equation}
where the first inequality follows from \eqref{eq:C1toC2}.

 Finally, define $C_3$ to be the cost incurred by the simulated copy of $F$. We have by construction that $\E[C_3] \le O(T \cdot G_F)$, and we have by Lemma \ref{lem:aux-moves} that the cost $C_4$ of deadweight moves is $O(T)$. Since $C_2 = O(C_3) + C_4$, it follows that $\E[C_2] \le O(T \cdot G_F)$. Therefore, by \eqref{eq:c2n1}, the embedding has amortized expected cost $O(C_2 / T + 1) = O(G_F)$ per operation.
 \end{proof}

Critically, if $R$ and $F$ \emph{both} offer lightly-amortized guarantees, then we can show that $\Layer{F}{R}$ \emph{also} offers lightly-amortized guarantees. This is what will allow for us to structure Theorem \ref{thm:twocomp} in such a way that it can be applied twice to obtain Theorem \ref{thm:threecomp}. 
 \begin{lemma}
Suppose $G_F(\overline{x})$ is the same value for all input sequences $\overline{x}$, and refer to this value as $G_F$. Suppose that $R$ has a lightly-amortized expected cost $O(E_R)$ per operation; and that $F$ has a lightly-amortized expected cost $O(G_F(\overline{x}))$ per operation. Then, on any input sequence $\overline{x}$ of length $\Omega(n)$, the embedding $\Layer{F}{R}$ incurs \textbf{lightly-}amortized expected cost $O(G_F(\overline{x}))$ per operation.
\label{lem:EF2}
\end{lemma}

It is worth remarking on why, in Lemma \ref{lem:EF2}, $G_F(\overline{x})$ takes the same value $G_F$ for all $\overline{x}$. This is so that it makes sense to talk about $F$ offering a \emph{lightly}-amortized expected cost of $G_F$. This means that, on any input subsequence $\overline{x}'$ of some length $T$, the total expected cost on that subsequence should be at most $O(T G_F + n)$. If we want $G_F$ to depend on $\overline{x}$, then we would need to make it also depend on the \emph{subsequence} $\overline{x}'$. One could extend Lemma \ref{lem:EF2} to this setting, but as we shall see later on, it is not actually necessary for our results.

\begin{proof}
    The proof is similar to that of Lemma \ref{lem:EF1}, except that now we focus on an input \emph{sub}sequence $\hat{x} = \overline{x}_t, \ldots, \overline{x}_{t + T}$ for some $t$ and $T$. 

    By the same reasoning as for \eqref{eq:c2n1} in Lemma \ref{lem:EF1}, we have that the total expected cost incurred by the embedding on $\hat{x}$ is
    \begin{equation}O(C_2 + n),
    \label{eq:c2n2}
    \end{equation}
    where $C_2$ is the total expected cost incurred by the $F$-emulator on $\hat{x}$. 
    
    Define $C_3$ to be the cost incurred by the simulated copy of $F$ on $\hat{x}$. We have by light amortization that $\E[C_3] \le O(T \cdot G_F + n)$. We also have by Lemma \ref{lem:aux-moves} that $\E[C_2 - C_3]$ (i.e., the cost of deadweight moves) is $O(T + S)$, where $S$ is the maximum size of any rebuild batch. This implies by Lemma \ref{lem:checkpoint-length} that $\E[C_2 - C_3] \le O(T + n)$, implying that $\E[C_2] \le O(T + n) + \E[C_3] \le O(T \cdot G_F + n)$. By \eqref{eq:c2n2}, it follows that the total expected cost incurred by the embedding on $\hat{x}$ is $O(T \cdot G_F + n)$.
\end{proof}

Finally, we show that $\Layer{F}{R}$ also enjoys the cost guarantees offered by $R$.
\begin{lemma}
Suppose that $R$ has lightly-amortized expected cost $O(E_R)$ per operation and that $R$ has a worst-case cost $O(W_R)$ per operation, where $W_R \ge E_R$. Then the embedding $\Layer{F}{R}$ also has lightly-amortized expected cost $O(E_R)$ per operation and worst-case cost $O(W_R)$ per operation.\label{lem:Rcost}
\end{lemma}
\begin{proof}
    The main claim that we must establish is that, during each operation, the cost that we incur on the $F$-emulator (including deadweight moves) is always at most $O(E_R)$. This is immediate for operations that trigger the slow path, as the only work that they perform on the $F$-emulator is in steps (i) and (iv) of part (b) of the slow path. To analyze operations that take the fast path, recall that the reason an operation goes to the fast path is that (1) its cost $c_E$ in the simulated copy of $F$ is at most $E_R$ and (2) there is no pending rebuild work. The lack of pending rebuild work means that there are no buffered elements, so the operation will not incur any cost amplification. Thus the operation incurs true cost $c_E \le E_R$. 
    
    Since each operation performs at most $O(E_R)$ work on the $F$-emulator, it remains only to consider the cost incurred on the $R$-shell. By Lemma \ref{lem:nocycles}, the cost that we incur on the $R$-shell is bounded by $R$'s guarantees---$O(E_R)$ in expectation and $O(W_R)$ in the worst case. This implies the lemma.
\end{proof}

Putting the pieces together, we are ready to prove Theorem~\ref{thm:twocomp}.

\restatablethmone*
\begin{proof}
We must first establish that $\Layer{F}{R}$ is a valid list-labeling algorithm. 
The embedding $\Layer{F}{R}$ keeps all elements in sorted order by construction.
Thus the only opportunity for incorrectness is if the algorithm is impossible to execute, that is, if in step (i) of part (a) of the slow path, there are no more buffer dummy elements left. 
But, we know by Lemma \ref{lem:buffer-slots-free} that this never happens.

The cost guarantees follow directly from Lemmas \ref{lem:EF1}, \ref{lem:EF2}, and \ref{lem:Rcost}.
\end{proof}

Applying Theorem \ref{thm:twocomp} twice, we can immediately obtain Theorem \ref{thm:threecomp}. Notice that, although light amortization does not appear at all in the specifications of Theorem \ref{thm:twocomp}, the concept ends up being critical to the proof, as it allows for the output of the first embedding $\Layer{Y}{Z}$ to be reused as the input to the second $\Layer{X}{(\Layer{Y}{Z})}$.

\restatablethmtwo*
\begin{proof}
First, consider the embedding $\Layer{Y}{Z}$, which has $G_F = B$ and $W_R = E_R = C$. By Theorem \ref{thm:twocomp}, the embedding has amortized expected cost $O(B)$ and worst-case $O(C)$. Moreover, since $Y$'s guarantee is an expected-cost guarantee (not an amortized guarantee) the $O(B)$ guarantee for $\Layer{Y}{Z}$ is actually for \emph{lightly}-amortized expected cost. This final fact is critical, as it is what allows $\Layer{Y}{Z}$ to be reused as the input to another embedding, where it serves as $R$ with $W_R = C$ and $E_R = B$. 

Now consider the second embedding $\Layer{X}{(\Layer{Y}{Z})}$, which has $G_F = A(\overline{x})$, $E_R = B$, $W_R = C$. Again applying Theorem \ref{thm:twocomp}, we have that $\Layer{X}{(\Layer{Y}{Z})}$ simultaneously achieves amortized expected cost $O(A(\overline{x}))$ on any input $\overline{x}$; amortized expected cost $O(B)$ on any input; and $O(C)$ worst-case cost on any input.
\end{proof}

\begin{corollary}
There exists a list-labeling algorithm that, on any input sequence of length $\Omega(n)$, achieves:
\begin{itemize}
    \item Amortized expected cost $O(\log n)$ if the input sequence is a \emph{hammer-insert workload}, as defined in \cite{BenderHu07};
    \item Amortized expected cost $O(\log^{3/2}n)$;
    \item Worst-case cost $O(\log^2 n)$.
\end{itemize}
\end{corollary}

\begin{proof}
The result follows by applying Theorem~\ref{thm:threecomp} to the following three list-labeling algorithms:
\begin{itemize}
\item As $X$: The algorithm from \cite{BenderHu07}, which achieves amortized $O(\log n)$ cost on hammer-insert workloads.
\item As $Y$: The algorithm from \cite{BenderCoFa22}, which achieves expected cost $O(\log^{3/2} n)$ on all inputs.
\item As $Z$: The algorithm from \cite{Willard92}, which achieves deamortized $O(\log^2 n)$ cost on all inputs. 
\end{itemize}
\end{proof}

\begin{corollary}
Let $\overline{x} = \overline{x}_1, \overline{x}_2, \ldots, \overline{x}_n$ be a sequence of $n$ insertions, and let $\pi:[n]\rightarrow [n]$ be the permutation such that $\overline{x}_{\pi(1)} < \overline{x}_{\pi(2)} < \cdots < \overline{x}_{\pi(n)}$. Let $P:\{x_i\} \rightarrow [n]$ be a \emph{rank predictor}, and let $\eta = \max_i |\pi(i) - P(\overline{x}_i)|$ denote the maximum error that $P$ incurs across the insertions. Then there is an online (learning-augmented) list-labeling algorithm that, when equipped with $P$, supports the insertion sequence $\overline{x}$ with:
\begin{itemize}%lkjlkjlkj
    \item Amortized expected cost $O(\log^2 \eta)$;
    \item Amortized expected cost $O(\log^{3/2}n)$;
    \item Worst-case cost $O(\log^2 n)$.
\end{itemize}
\end{corollary}

\begin{proof}
The result follows by applying Theorem~\ref{thm:threecomp} to the following three list-labeling algorithms:
\begin{itemize}
\item As $X$: The learning-augmented algorithm from \cite{predictions}, which achieves amortized $O(\log \eta^2)$ cost on $\overline{x}$.
\item As $Y$: The algorithm from \cite{BenderCoFa22}, which achieves expected cost $O(\log^{3/2} n)$ on all inputs.
\item As $Z$: The algorithm from \cite{Willard92}, which achieves deamortized $O(\log^2 n)$ cost on all inputs. 
\end{itemize}
\end{proof}